\pdfoutput=1
\documentclass[lettersize,journal]{IEEEtran}
\usepackage{amsmath,amsfonts,amssymb}
\usepackage{algorithmic}
\usepackage{algorithm}
\usepackage{array}
\usepackage[caption=false,font=normalsize,labelfont=sf,textfont=sf]{subfig}
\usepackage{textcomp}
\usepackage{url}
\usepackage{graphicx}
\usepackage{cite}
\usepackage{booktabs}
\usepackage{xcolor}
\usepackage{multirow}
\usepackage[hidelinks]{hyperref}
\usepackage{tikz}
\usetikzlibrary{arrows.meta,positioning,shapes.geometric,fit,calc,decorations.pathreplacing}
\usepackage{amsthm}
\newtheorem{theorem}{Theorem}

\newtheorem{definition}{Definition}

\newcommand{\sbpp}{\mbox{\textnormal{\textsc{SBPP}}}}

\newcommand{\LP}{\mathrm{LP}}
\newcommand{\negl}{\mathrm{negl}}

\begin{document}

\title{Search-Bound Proximity Proofs:\\
Binding Encrypted Geographic Search to\\
Zero-Knowledge Verification}

\author{%
Yoshiyuki Ootani%
\thanks{Yoshiyuki Ootani is an independent researcher based in Japan (e-mail: info@ootanl.com).}%
}

\maketitle

\begin{abstract}
Location-based systems that combine encrypted geographic search
with zero-knowledge proximity proofs typically treat the two
phases as independent. Under an honest-but-curious server, this
leaves an \emph{authorization provenance gap}: once session state
is purged, no forensic procedure can attribute a proof to its
originating search session, because the proof's public inputs
encode no session-identifying information.

We formalize this gap as the \emph{search-authorized proof}
(SAP) security notion and show via a concrete audit re-association
attack that proof-external mechanisms (where authorization evidence
remains outside the proof) cannot prevent forensic misattribution
when the same drop parameters recur across sessions.
Search-Bound Proximity Proofs (\sbpp{}) realize the SAP
requirements without modifying the ZKP circuit: session nonce,
Merkle-root result-set commitment, and signed receipt are
decomposed into independently auditable components, enabling
property-level fault isolation in offline audit. Experiments on
synthetic and real-world data (110\,776 OpenStreetMap POIs) show
sub-millisecond absolute overhead on a 125\,ms Groth16 baseline.
\end{abstract}
\begin{center}
\small This work has been submitted to the IEEE for possible publication.
Copyright may be transferred without notice, after which this version may no longer be accessible.
\end{center}


\section{Introduction}
\label{sec:intro}

\IEEEPARstart{G}{eo-content} systems---applications that bind digital
content to physical locations---are a growing class of location-based
services (LBS). Examples include location-locked media (``geo-drops''),
geofenced access control, and place-based augmented reality. These
systems share a common two-phase workflow:

\begin{enumerate}
\item \textbf{Discovery.} The client queries the server to find nearby
  content (e.g., ``what drops are within 1\,km of me?'').
\item \textbf{Verification.} The client proves proximity to a specific
  piece of content to unlock it (e.g., via GPS check, zero-knowledge
  proof, or visual verification).
\end{enumerate}

In current deployments, these phases are \emph{decoupled}: discovery
uses plaintext geohash queries that reveal the user's approximate
location to the server, while verification uses a separate
cryptographic mechanism (such as a Groth16 proximity proof) that reveals
nothing about the user's exact coordinates. This creates what we call
the \emph{search-verify gap}---more precisely, an
\emph{authorization provenance gap}: once session state is purged,
no audit record can forensically attribute a proof to its
originating search session, because the proof's public inputs
encode no session-identifying information.

\smallskip
Prior work has addressed parts of this problem in isolation.
Encrypted geographic search systems~\cite{gridse,hier-sse,range-sse}
replace plaintext geospatial queries with encrypted or tokenized search
mechanisms that hide the user's coordinates from the server.
Pairing-based zero-knowledge arguments such as Groth16~\cite{groth16},
when instantiated as a proximity circuit~\cite{zairn2026}, allow the
user to prove they are within a radius~$R$ of a target without revealing
their exact location. However, to our knowledge, prior work has
not studied binding \emph{encrypted geographic discovery} to
\emph{zero-knowledge proximity verification} via the proof
transcript, leaving the authorization provenance gap open.

A natural first attempt is to check a session nonce alongside the
proof at the application layer. We show (\S\ref{sec:eval}) that
this \emph{does not suffice}: because the nonce is not committed
inside the proof, the proof can be separated from its session
context and re-associated with a different session (A1/A2).
Stronger alternatives such as signed authorization tokens or
MAC-based result binding also fall short: because their proofs
remain session-agnostic, an adversary can re-associate a valid
proof with a different session's token, forging a misleading audit
record (\S\ref{sec:security:separation}).
Under an honest-but-curious server, this matters for
\emph{forensic accountability}: in geo-content platforms subject to
compliance audit (e.g., geofenced access control, location-gated
media), post-hoc review must attribute each proof to its
originating authorization. If an adversary swaps a proof's
authorization token, the audit record misattributes the access to
a different session. Log-level correlation is heuristic and
defeatable; only cryptographic commitment in the proof's public
inputs makes misattribution detectable.

We formalize this as the SAP security notion
(\S\ref{sec:security}) and present \sbpp{} as a practical
decomposition: a session nonce (P1), a Merkle-root
digest (P2), and a signed receipt (P3). Any internalized-token
design with equivalent content (V8 in our comparison) achieves
the same security level; \sbpp{}'s advantage is property-level
fault isolation and format-independent audit
(\S\ref{sec:discussion}). \sbpp{} certifies provenance of the \emph{issued} context, not
search correctness (\S\ref{sec:discussion}).

\smallskip\noindent\textbf{Contributions.}

\begin{itemize}
\item[\textbf{C1}] We define the \emph{search-authorized proof}
  (SAP) security notion (P1--P3), identify six attack classes
  including forensic audit misattribution (A4b), and establish
  design requirements for the committed content required by each
  property (\S\ref{sec:problem}, \S\ref{sec:security}).
\item[\textbf{C2}] We exhibit an \emph{audit re-association attack}
  showing that proof-external authorization is vulnerable to P3
  violation whenever the same drop parameters recur across
  sessions---a common scenario within epoch windows
  (\S\ref{sec:security:separation}).
\item[\textbf{C3}] We present \sbpp{}, a deployable realization
  of the SAP requirements that achieves property-level fault
  isolation (each failure mode maps to a distinct component),
  $O(1)$ compact verifier state, and no circuit modification
  (\S\ref{sec:protocol}).
\item[\textbf{C4}] We evaluate on synthetic and real-world (OSM)
  distributions, nine protocol variants, and two mobile devices
  (\S\ref{sec:eval}).
\end{itemize}

\noindent\textbf{Prior artifact and novel contributions.}
The Zairn platform~\cite{zairn2026} provides three pre-existing
components used in this work: (i)~a GridSE-style encrypted search
module, (ii)~a Groth16 proximity-proof circuit (474 constraints),
and (iii)~a proof-level context binding mechanism that embeds
drop ID, policy version, and epoch in the challenge digest to
prevent cross-drop transfer (A3).

\emph{New to this paper:}
the SAP security notion (P1--P3) and design requirements; the
audit re-association attack and proof-external separation; the
\sbpp{} practical decomposition (Core/Full); and the nine-variant
evaluation including OSM real-world data. The prior
artifact~\cite{zairn2026} provides implementation components
(SSE module, Groth16 circuit, context binding); the security
notion, separation argument, and full evaluation are introduced
here.

\section{Background}
\label{sec:background}

\subsection{Searchable Symmetric Encryption}

Searchable symmetric encryption (SSE) enables keyword search over
encrypted data~\cite{song-sse,structured-encryption}. In the geographic
setting, GridSE~\cite{gridse} encodes locations as geohash strings and
generates HMAC-based index tokens at multiple precision levels. A
searcher generates tokens for their location and its neighbors; the
server performs set intersection over opaque tokens without learning
coordinates.

The standard SSE leakage profile includes: (i)~\emph{search pattern}
(whether two queries are identical), (ii)~\emph{access pattern}
(which documents match), and (iii)~\emph{volume pattern}
(how many documents match)~\cite{sse-leakage}. Geographic SSE
inherits all three.

\subsection{Zero-Knowledge Proximity Proofs}

A zero-knowledge proximity proof (ZKPP) allows a prover to convince a
verifier that they are within radius~$R$ of a target location
$(t_\text{lat}, t_\text{lon})$ without revealing their exact
coordinates~\cite{groth16}. In Zairn's implementation~\cite{zairn2026},
the prover generates a Groth16 proof over a Circom circuit (474
constraints on BN128) with public inputs including the target
coordinates, radius, and a \emph{challenge digest}---a hash that
binds the proof to a specific application context (drop ID, policy
version, epoch).

\subsection{Context Binding}

The Zairn platform~\cite{zairn2026} implements context binding for ZKPPs:
embedding application-level identifiers (drop ID, policy version, epoch)
in the proof's public inputs to prevent cross-drop proof transfer. The
challenge digest is computed as:
\begin{equation}
\text{contextDigest} = H(\LP(D, \text{pv}, e))
\label{eq:context}
\end{equation}
where $\LP(\cdot)$ denotes length-prefixed canonical encoding and $H$
is SHA-256 reduced modulo the BN128 scalar field order $q$.

\sbpp{} extends this by adding the search session nonce~$N$ to the
encoding (\S\ref{sec:protocol:binding}). Crucially, the prior
proof-level binding~\cite{zairn2026} addresses only \emph{proof-level} binding
(which drop a proof is for) and does not consider the discovery
phase at all. \sbpp{} addresses \emph{protocol-level} binding
(which search session a proof belongs to), a complementary and
non-overlapping concern.

\section{Problem Statement and Threat Model}
\label{sec:problem}

\subsection{System Model}

We consider a geo-content system with three roles:
\begin{itemize}
\item \textbf{Drop creator}: Publishes location-bound content with
  encrypted index tokens and a ZKP verification key.
\item \textbf{Client} (searcher/prover): Searches for nearby drops
  and generates proximity proofs to unlock them. Holds a shared
  search key $k$ for HMAC token generation (provisioned per
  content-access tier; key distribution is orthogonal to \sbpp{}).
\item \textbf{Server}: Stores encrypted index tokens and drop metadata.
  Performs token matching and proof verification.
\end{itemize}

\subsection{Adversary Model}

We consider two adversary types:

\smallskip\noindent\textbf{Privacy adversary (honest-but-curious
server).} The server follows the protocol correctly but attempts to
learn the client's location from observed messages---search tokens,
access patterns, volume patterns, and proof submission
timing~\cite{sse-leakage}. This is the standard SSE adversary
model~\cite{curtmola-sse,sse-leakage}.

\smallskip\noindent\textbf{Integrity adversary (active client /
network attacker).} An adversary who can intercept, replay,
substitute, or re-associate protocol messages. This adversary
attempts cross-session proof substitution (A1), re-association (A2),
result-set escape (A4), or delayed replay (A5). Cross-drop transfer
(A3) is handled by the underlying context
binding~\cite{zairn2026}.

\smallskip
\sbpp{} primarily targets the integrity adversary under an
HBC server. A \emph{malicious} server can omit candidates, sign
biased roots, or refuse sessions; \sbpp{} does not guarantee
search completeness or fairness (see \S\ref{sec:discussion}).
Per-query SSE leakage is unchanged from GridSE
(\S\ref{sec:security}).

\subsection{The Search-Verify Gap}

\begin{definition}[Authorization Provenance Gap]
A geo-content protocol has an authorization provenance gap if,
after server-side session state is purged, no forensic procedure
can cryptographically attribute a proof to the search session that
authorized it---i.e., the proof's committed public inputs encode
no session-identifying information.
\end{definition}

\subsection{Attack Classes}

The authorization provenance gap enables the following attack classes:

\smallskip\noindent\textbf{A1: Cross-session proof substitution.}
A proof generated in session $S_1$ is submitted in $S_2$. Without
proof-level session binding, the server cannot detect this from the
proof transcript alone.

\smallskip\noindent\textbf{A2: Cross-session re-association.}
An attacker intercepts a legitimate proof and submits it under a
different session.

\smallskip\noindent\textbf{A3: Cross-drop transfer.}
A proof for drop $D_1$ is submitted for $D_2$. Prevented by context
binding~\cite{zairn2026} (drop ID in digest).

\smallskip\noindent\textbf{A4a: Result-set escape (online).}
A proof for a drop $D \notin \mathcal{R}$ is submitted. The
deployed verifier catches this via an explicit membership check
$D \in \mathcal{R}$, but this requires the server to retain
$\mathcal{R}$ in session state.

\smallskip\noindent\textbf{A4b: Result-set escape (audit).}
After server-side $\mathcal{R}$ is purged, can an auditor verify
from the stored audit record that the proof was authorized for a
drop in the result set? Without a committed root in $\text{pub}$
and authenticated openings, no.

\smallskip\noindent\textbf{A5: Delayed proof replay.}
A client caches search results, waits, then submits a proof. Session
TTL and consumption mitigate this.

\section{\sbpp{} Protocol}
\label{sec:protocol}

\subsection{Protocol Overview}

Fig.~\ref{fig:protocol} shows the \sbpp{} message sequence.

\begin{figure}[!t]
\centering
\resizebox{\columnwidth}{!}{%
\begin{tikzpicture}[
  >=Stealth,
  entity/.style={font=\footnotesize\bfseries, minimum width=1.2cm},
  msg/.style={font=\scriptsize, above, midway},
  msgb/.style={font=\scriptsize, below, midway},
  stepnum/.style={font=\scriptsize\itshape, text=gray!60!black},
]
\node[entity] (C) at (0, 0) {Client};
\node[entity] (S) at (6, 0) {Server};

\draw[dashed, gray] (C.south) -- ++(0, -7.5);
\draw[dashed, gray] (S.south) -- ++(0, -7.5);

\draw[->, thick] (0, -0.8) -- node[msg] {\texttt{initSession()}} (6, -0.8);
\node[stepnum] at (-1.5, -0.8) {1};

\draw[<-, thick] (0, -1.5) -- node[msg] {$S, N, t_{\exp}$} (6, -1.5);
\node[stepnum] at (7.2, -1.5) {2};

\node[stepnum, align=left] at (-1.5, -2.4) {3};
\node[font=\scriptsize, align=left] at (0.2, -2.4) {$T \gets \text{HMAC}(k,$\\$\;\text{``gridse:}\mathit{p}\text{:}\mathit{gh}\text{''})$};

\draw[->, thick] (0, -3.3) -- node[msg] {$T$ (opaque tokens)} (6, -3.3);
\node[stepnum] at (-1.5, -3.3) {4};

\draw[<-, thick] (0, -4.0) -- node[msg] {candidate drop IDs} (6, -4.0);
\node[stepnum] at (7.2, -4.0) {5};

\node[stepnum, align=left] at (-1.5, -5.0) {6--7};
\node[font=\scriptsize, align=left] at (0.4, -5.1) {select drop $D$\\$cd \gets H(\LP(\ldots,N,\mathit{root}))$\\generate $\pi$ + Merkle path};

\draw[->, thick] (0, -6.2) -- node[msg] {$\pi$, public signals, $S$} (6, -6.2);
\node[stepnum] at (-1.5, -6.2) {8};

\node[font=\scriptsize, align=left] at (6.3, -6.9) {verify $\pi$\\check $N = N_S$};
\node[stepnum] at (7.7, -6.8) {9};

\draw[<-, thick] (0, -7.3) -- node[msg] {grant / deny} (6, -7.3);
\node[stepnum] at (7.2, -7.3) {10};

\end{tikzpicture}}
\caption{\sbpp{} protocol. After token matching (step~5), the server
computes a Merkle root over the result set (step~6). The client
embeds both $N$ and the root in the ZKP challenge digest (step~8)
and submits a Merkle membership proof for the selected drop.
Search tokens (step~3) do not contain $N$.}
\label{fig:protocol}
\end{figure}
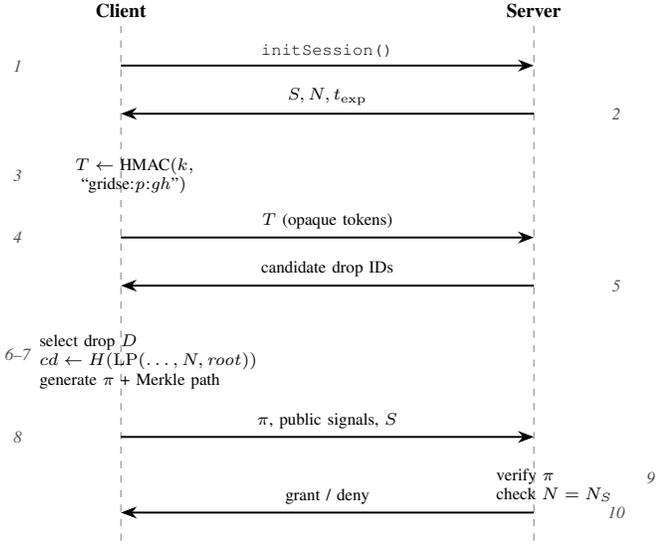

The protocol proceeds as follows:

\begin{enumerate}
\item Client requests a search session.
\item Server generates a cryptographically random nonce $N$,
  session ID $S$, and expiry time $t_{\exp}$. Stores $(S, N, t_{\exp})$.
\item Client generates encrypted search tokens
  $T = \{\text{HMAC}(k, \texttt{gridse:index:}p\texttt{:}g_i)\}$
  for center geohash $g_0$ and neighbors $g_1, \ldots, g_8$ at
  precision~$p$ selected for the search radius.
\item Client sends $T$ to the server.
\item Server performs set intersection against stored index tokens
  and returns matching drop IDs with metadata.
\item Server builds a binary Merkle tree over sorted candidate IDs
  $\{D_1, \ldots, D_k\}$ and stores only the Merkle root
  $\mathit{root}(\mathcal{R})$ in the session record.
\item Client selects a drop $D$ to unlock.
\item Client computes the \sbpp{} challenge digest:
  \begin{equation}
  cd = H(\LP(\texttt{SBPP-v1}, D, \text{pv}, e, N,
    \mathit{root}(\mathcal{R})))
  \label{eq:sbpp-digest}
  \end{equation}
  and generates a Groth16 proximity proof $\pi$ with $cd$ as a
  public input (\texttt{pub[7]}).
\item Client sends $\pi$, public signals, session ID $S$, and a
  Merkle membership proof for $D$ to the server.
\item Server verifies: (a)~$S$ is a valid, non-expired session with
  nonce $N$ and root $\mathit{root}(\mathcal{R})$;
  (b)~the proof's challenge digest matches Eq.~\ref{eq:sbpp-digest};
  (c)~the Merkle membership proof for $D$ is valid against
  $\mathit{root}(\mathcal{R})$;
  (d)~the Groth16 proof $\pi$ is valid.
\item Server returns grant (and consumes the session) or deny.
\end{enumerate}

\noindent Unless otherwise noted, the steps above describe
\emph{Full \sbpp{} in compact-online mode}, where the verifier
stores only a Merkle root and checks a membership witness.
\emph{Core \sbpp{}} is obtained by omitting
$\mathit{root}(\mathcal{R})$ from the challenge digest, omitting
the Merkle witness, and retaining $\mathcal{R}$ server-side for
the explicit online membership check (V4a).

\subsection{Session Management}

Sessions are server-side state with three properties:
\begin{itemize}
\item \textbf{Uniqueness}: Nonce $N$ is 256 bits of cryptographic
  randomness (collision probability $< 2^{-128}$ after $2^{64}$
  sessions).
\item \textbf{Expiry}: Default TTL of 5 minutes. Expired sessions
  reject all operations.
\item \textbf{Consumption}: A session is consumed (deleted) after
  one successful proof verification. This prevents proof replay
  within the same session.
\end{itemize}

Server-side state depends on the verifier mode. In Core/stateful
mode, the server stores $(N, t_{\exp}, \mathcal{R})$ and enforces
online authorization via $D \in \mathcal{R}$; state is
$O(|\mathcal{R}|)$ per session. In Full/compact mode, the server
stores $(N, t_{\exp}, \mathit{root}(\mathcal{R}))$---a single hash
value regardless of $|\mathcal{R}|$; state is $O(1)$ per session.
In both modes, total active sessions are bounded by the concurrent
user count (5-minute TTL, one-time consumption).

\subsection{Search-Proof Binding}
\label{sec:protocol:binding}

The binding between search and verification is achieved through
the nonce $N$:

\begin{itemize}
\item The search phase operates within session $S$ associated
  with nonce $N$.
\item The proof phase includes $N$ in the challenge digest
  (Eq.~\ref{eq:sbpp-digest}), which becomes a public input to
  the Groth16 circuit.
\item The server verifies that the proof's challenge digest
  matches the nonce stored for session $S$.
\end{itemize}

Because the challenge digest is a SHA-256 hash over
length-prefixed fields including~$N$, changing~$N$ changes the
digest with overwhelming probability (collision resistance).
The Groth16 proof commits to the digest as a public input, so
it cannot be modified post-generation without invalidating the proof
(soundness of Groth16).

\subsection{Digest-Based Binding}

The existing circuit~\cite{zairn2026} accepts a challenge digest
as public input \texttt{pub[7]} but does \emph{not} recompute it
internally. The authorization context $(N, \mathit{root}, D,
\text{pv}, e)$ is bound to the proof through a \emph{commitment
chain}: (1)~the prover computes $cd = H(\LP(\ldots))$ over the
full authorization context and supplies it as \texttt{pub[7]};
(2)~Groth16 soundness ensures $\pi$ is valid only for the
committed \texttt{pub[7]}; (3)~the verifier independently
recomputes $H(\LP(\ldots))$ from its own session record and
checks equality with \texttt{pub[7]}; (4)~SHA-256 collision
resistance and LP injectivity ensure that equality implies
agreement on $N$, $\mathit{root}$, $D$, $\text{pv}$, and $e$.
Thus, the authorization context is \emph{instance-bound} to the
proof via \texttt{pub[7]}: the verifier independently recomputes
the digest and checks equality, so any context mismatch is
detected without circuit-internal processing. This preserves:
\begin{itemize}
\item No new trusted setup or verification key.
\item Circuit constraints unchanged (474).
\item Backward compatibility with non-SBPP flows.
\end{itemize}
\noindent A domain separator \texttt{SBPP-v1} as the first LP
field prevents downgrade attacks.

\subsection{Result-Set Binding}

Session binding (Eq.~\ref{eq:sbpp-digest}) ties the proof to a
session but not to the \emph{specific drops returned} in that
session. Full \sbpp{} closes this gap by committing a Merkle
root~\cite{merkle-signature} of the authorized result set inside
the challenge digest (and thus inside $\text{pub}[7]$).

The server builds a binary Merkle tree over sorted candidate IDs
and stores only the root $\mathit{root}(\mathcal{R})$. The client
recomputes the tree locally and embeds the root in
Eq.~\ref{eq:sbpp-digest}. When unlocking drop $D$, the client
also submits a Merkle membership proof. This binds the proof to
(i)~session ($N$), (ii)~drop ($D$), and (iii)~result set
($\mathit{root}(\mathcal{R})$). Forging requires breaking Merkle
or Groth16 security.

\subsection{Verifier Modes}

Full \sbpp{} supports three modes: (1)~\emph{Stateful online}:
server stores $\mathcal{R}$, checks $D \in \mathcal{R}$ directly
($O(|\mathcal{R}|)$ state). (2)~\emph{Compact online}: server
stores only the Merkle root ($O(1)$ state); client submits a
membership witness. (3)~\emph{Offline audit}: after state is
purged, an auditor re-verifies from $\tau_\text{audit} =
(\pi, \text{pub}, D, \mathit{merklePath}, \rho)$---the committed
$\text{pub}[7]$ plus authenticated openings. Core \sbpp{} supports
(1) only; Full supports all three.

\subsection{Receipt-Backed Offline Audit}

To satisfy P3 after session state is purged, the server issues a
\emph{signed session receipt}
$\rho = \text{Sign}_\text{srv}(S, N, t_{\exp},
\mathit{root}(\mathcal{R}), \text{mode}, \text{pv}, e)$
after the result set and Merkle root are computed. The audit log
retains
$(\rho, D, \mathit{merklePath}, \text{pub}, \pi)$, where $\pi$ is
the full Groth16 proof ($\sim$724\,B)---retaining $\pi$ (not only
its hash) is necessary for independent re-verification. An auditor
replays authorization via four checks: (i)~server signature on
$\rho$; (ii)~recomputation of $cd_\text{full}$ from receipt
fields; (iii)~Merkle membership of $D$; (iv)~Groth16 proof
verification. The proof's $\text{pub}[7]$ commits the
authorization context; $\rho$ and $\mathit{merklePath}$ serve as
\emph{authenticated openings} that an auditor checks against that
instance binding, yielding authorization provenance (P3)
independent of mutable session state. The audit log may be held by
the server operator or an independent compliance party; the auditor
needs only the server's public signing key---no session state,
MAC secrets, or token format knowledge.

\section{Security Analysis}
\label{sec:security}

\subsection{Search-Authorized Proofs}

A search-authorized proof scheme consists of five algorithms:
\begin{itemize}
\item $\text{Issue}(1^\lambda) \to (S, N, t_{\exp})$: server
  creates a session with random nonce.
\item $\text{Search}(S, T) \to (\mathcal{R},
  \mathit{root}(\mathcal{R}), \rho)$: server matches tokens,
  builds Merkle tree, signs receipt.
\item $\text{Prove}(D, N, \mathit{root}, w) \to (\pi,
  \text{pub})$: client generates ZKP with $cd$ in pub[7].
\item $\text{Verify}(S, \pi, \text{pub}, D,
  \mathit{merklePath}) \to \{0,1\}$: server checks session,
  digest, Merkle witness, and proof.
\item $\text{Audit}(\pi, \text{pub}, D, \mathit{merklePath},
  \rho) \to \{0,1\}$: offline re-verification using receipt
  signature, digest recomputation, Merkle membership, and Groth16
  verification---no live session state.
\end{itemize}

\noindent We define three security properties. P1 and P2 are
binding/soundness requirements; P3 captures whether an adversary
can produce a \emph{misleading} audit record.

\begin{definition}[Search-Authorized Proof]
\label{def:sap}
Let $\Pi$ be a NIZK proof system with public inputs and
$\mathcal{Q}$ a search protocol that returns a candidate set
$\mathcal{R}$ within a session~$S$. A \emph{search-authorized
proof} (SAP) \emph{scheme} for $(\Pi, \mathcal{Q})$ satisfies:

\smallskip\noindent\textbf{P1 (Authorization binding).}
A valid proof for drop $D$ under session $S$ cannot be accepted
under any $S' \neq S$.

\smallskip\noindent\textbf{P2 (Result-set soundness).}
A valid proof for $D \notin \mathcal{R}_S$ cannot be accepted
under $S$.

\smallskip\noindent\textbf{P3 (Authorization provenance).}
Consider an adversary $\mathcal{A}$ that interacts with oracles
$\mathcal{O}_\text{Issue}$, $\mathcal{O}_\text{Search}$,
$\mathcal{O}_\text{Prove}$ (obtaining honest session transcripts
and proofs), and then outputs a candidate audit record
$\tau^*_\text{audit} = (\pi^*, \text{pub}^*, D^*,
\mathit{mp}^*, \rho^*)$. $\mathcal{A}$ wins if
$\text{Audit}(\tau^*_\text{audit}) = 1$ but $\pi^*$ was
generated by $\mathcal{O}_\text{Prove}$ under a different
authorization context $(N', \mathit{root}')$ from the one
claimed by $\rho^*$---i.e., $N' \neq N_{\rho^*}$ or
$\mathit{root}' \neq \mathit{root}_{\rho^*}$. P3 requires
$\Pr[\mathcal{A} \text{ wins}] \leq \negl(\lambda)$.
\end{definition}

\noindent We distinguish the \emph{proof transcript}
$\tau = (\pi, \text{pub})$ from the \emph{audit transcript}
$\tau_\text{audit}$, which additionally includes the Merkle
witness and signed receipt.

\sbpp{} instantiates this definition: Core satisfies P1,
Full satisfies P1+P2, and Full with signed receipts satisfies
P1+P2+P3. Before showing how, we first establish \emph{what must
be committed}---a design requirement that constrains any SAP scheme,
not just \sbpp{}.

\subsection{Necessity of Committed Content}

\noindent\textbf{Observation 1 (Formal Requirements).}
\label{thm:necessity}
For any SAP scheme built over a sound NIZK:
\begin{itemize}
\item \textbf{P1 requires session-distinguishing pub.}
  If $\text{pub}(S_1, D) = \text{pub}(S_2, D)$ for $S_1 \neq S_2$
  with the same $(D, \text{pv}, e)$, then $\pi$ valid under $S_1$
  is also valid under $S_2$, violating P1.
\item \textbf{P3 requires server authentication of committed
  values.} pub is client-supplied; without an authenticated
  artifact (receipt), a client can fabricate $(N, \mathit{root})$
  and pass audit.
\end{itemize}

\noindent\textbf{Remark (Design Rationale for P2).}
P2 under compact verification ($o(n)$ verifier state) strongly
suggests a result-set commitment in pub. With $2^n$ possible
result sets and $o(n)$ bits of state, the verifier cannot
distinguish all membership queries from state alone (pigeonhole).
A commitment in pub checked via a witness is the standard approach;
we state this as a design rationale rather than a formal
impossibility.

\smallskip\noindent\textit{Justification of Observation~1.}
\textbf{P1}: By NIZK soundness, $V(\text{pub}, \pi)$ is
deterministic. If pub is session-independent, any $\pi$ valid
under $S_1$ is accepted under $S_2$---violating P1.
\textbf{P3}: pub is prover-supplied. Without server
authentication, the prover can fabricate $(N, \mathit{root})$ and
produce a valid proof and Merkle witness for any self-chosen
$\mathcal{R}$, passing audit with fabricated context.

\noindent\textbf{Implications.} These design requirements
constrain any SAP scheme: P1 and P3 impose formal requirements
(provably necessary), while P2 imposes a strong design
requirement under compact verification. \sbpp{}
instantiates them with minimal constructions: a random nonce, a
Merkle root ($O(1)$ state), and a signed receipt. The separation below shows that proof-external designs, which omit
the first two ingredients from pub, are vulnerable to audit
re-association.

\subsection{Separation: Proof-External Authorization}
\label{sec:security:separation}

\begin{definition}[Proof-external authorization]
An authorization mechanism is \emph{proof-external} if the
authorization evidence (capability token, MAC, etc.) is verified
alongside the NIZK proof but is \emph{not committed as a public
input} to the proven statement.
\end{definition}

\noindent\textbf{Proposition (Audit Re-Association Attack).}
For proof-external authorization schemes in which the same
$(D, \text{pv}, e)$ can appear across sessions, we construct
an attack in which the adversary produces a misleading audit
record that passes verification yet misattributes the proof to a
different session.

\smallskip\noindent\textit{Attack construction.}
Suppose an adversary observes two honest sessions $S_1, S_2$ for
the same drop $D$ (same pv, epoch). In a proof-external scheme,
both sessions produce proofs $\pi_1, \pi_2$ whose public inputs
agree on all session-independent fields (since pub does not encode
the session). The adversary
constructs a forged audit record by pairing $\pi_1$ (generated
under $S_1$) with the authorization evidence from $S_2$ (e.g., a
signed capability token $\sigma_2$ or receipt $\rho_2$). This
record passes audit: $\sigma_2$ verifies (valid signature),
$\pi_1$ verifies (valid proof for the same pub), yet the proof was
generated under $S_1 \neq S_2$.

\smallskip\noindent\textit{Precondition.}
This is a \emph{conditional} separation: the attack requires two
sessions for the same $(D, \text{pv}, e)$. Epoch rotation reduces
the eligible pool; the attack is most potent within an epoch
window---a realistic scenario for frequently accessed content.

Under \sbpp{}, this attack fails: $\pi_1$ commits to
$cd(S_1, D)$ which embeds $N_{S_1}$, while the receipt from
$S_2$ contains $N_{S_2} \neq N_{S_1}$, so the digest
recomputation in audit step~(ii) produces a mismatch.

\smallskip\noindent\textbf{Scope and limitations.}
This separation concerns \emph{authorization provenance}, not
online authorization. V5/V7 provide effective online binding and
may be entirely sufficient for deployments without post-hoc audit.

\smallskip\noindent\textbf{Why log-level mitigations are
insufficient.}
One might argue that enriching the audit log with heuristic
metadata (timestamps, IPs, correlation IDs) could prevent
re-association. We show this is insufficient in principle.

The NIZK verification function $V(\text{pub}, \pi)$ is
deterministic and depends \emph{only on pub}. For a
session-agnostic proof, $V$ returns 1 for any session sharing
$(D, \text{pv}, e)$. No external data---regardless of its
richness---can change $V$'s output or create a
\emph{cryptographic} link between $\pi$ and a specific session.
An adversary who holds two valid $(\pi_i, \sigma_i)$ pairs can
swap them; the audit log sees two individually valid records
($V$ accepts both, both signatures verify) with no
cryptographically detectable inconsistency. Heuristic correlators
(e.g., timestamp proximity) are probabilistic, can be defeated by
timing manipulation, and do not compose into a verifiable proof of
provenance.

This is exactly the gap that design requirement~(1)
identifies: \emph{session-distinguishing information must be in
pub} for any SAP scheme. Log enrichment operates outside pub and
therefore cannot substitute.
Table~\ref{tab:attack} confirms: V5--V7 pass A1--A4a but fail
A4b.

\subsection{Sufficiency: \sbpp{} Satisfies P1--P3}

We now show that \sbpp{} satisfies the properties identified
by the design requirements above. We write
$N_S$ for the session nonce, $\mathit{root}_S$ for the committed
Merkle root, and define:
\begin{align}
cd_\text{core}(S,D) &= H(\LP(\texttt{SBPP-v1}, D, \text{pv}, e, N_S))
\label{eq:cd-core} \\
cd_\text{full}(S,D) &= H(\LP(\texttt{SBPP-v1}, D, \text{pv}, e, N_S, \mathit{root}_S))
\label{eq:cd-full}
\end{align}

\subsection{Session Binding Soundness}

\begin{definition}[Core \sbpp{} Verification]
$\text{Verify}_\text{core}(S, \tau)$ accepts transcript
$\tau = (\pi, \text{pub}, D)$ under session $S$ iff:
(i)~$S$ is valid and non-expired;
(ii)~$\text{pub}[7] = cd_\text{core}(S, D)$;
(iii)~$\pi$ is a valid Groth16 proof over $\text{pub}$.
\end{definition}

\begin{theorem}[Core Session Binding]
\label{thm:core-binding}
If SHA-256 is collision-resistant and Groth16 is sound, then for
any PPT adversary $\mathcal{A}$ that generates transcript $\tau^*$
using a nonce $N^* \neq N_S$:
\[
\Pr[\text{Verify}_\text{core}(S, \tau^*) = 1] \leq \negl(\lambda).
\]
\end{theorem}

\begin{proof}[Proof sketch]
$\text{Verify}_\text{core}$ requires $\text{pub}[7] =
cd_\text{core}(S,D) = H(\LP(\ldots, N_S))$. Since $\LP$ is
injective, $\LP(\ldots, N^*) \neq \LP(\ldots, N_S)$ for
$N^* \neq N_S$. By Groth16 soundness, $\text{pub}[7]$ is the
digest actually committed in $\pi$. For acceptance, $\mathcal{A}$
must therefore find a SHA-256 collision, which succeeds only with
negligible probability.
\end{proof}

\subsection{Cross-Session Isolation}

\begin{theorem}[Cross-Session Isolation]
\label{thm:cross-session}
For sessions $S_1, S_2$ with $N_1 \neq N_2$, any transcript
$\tau_1$ accepted under $S_1$ is rejected under $S_2$:
\[
\Pr[\text{Verify}_\text{core}(S_2, \tau_1) = 1] \leq \negl(\lambda).
\]
\end{theorem}

\begin{proof}[Proof sketch]
$\tau_1$ commits to $cd_\text{core}(S_1, D)$ which embeds $N_1$.
Verification under $S_2$ expects $cd_\text{core}(S_2, D)$ with
$N_2 \neq N_1$. By injectivity of $\LP$ and collision resistance
of $H$, the digests differ except with negligible probability.
\end{proof}

\subsection{Transcript-Level Result-Set Authorization}

\begin{definition}[Full \sbpp{} Verification]
$\text{Verify}_\text{full}(S, \tau)$ extends Core verification
with: (iv)~the Merkle membership witness for $D$ is valid against
$\mathit{root}_S$ committed in $cd_\text{full}(S, D)$.
\end{definition}

\begin{theorem}[Transcript-Level Authorization (A4b)]
\label{thm:transcript-auth}
For any $D \notin \mathcal{R}$ (the candidate set committed by
$\mathit{root}_S$), no PPT adversary can cause
$\text{Verify}_\text{full}(S, \tau) = 1$, assuming Merkle
target-binding security (no PPT adversary can produce a valid path
for $D \notin \mathcal{R}$ against a committed root), instantiated
via SHA-256.
\end{theorem}

\begin{proof}[Proof sketch]
$cd_\text{full}(S, D)$ includes $\mathit{root}_S$. Acceptance
requires a valid Merkle witness that reconstructs $\mathit{root}_S$
from the leaf $H(\LP(\texttt{SBPP-LEAF}, D))$. For
$D \notin \mathcal{R}$, producing such a witness requires either
forging a membership path against the committed root or finding a
second-preimage/collision in the Merkle hash construction, both of
which succeed only with negligible probability.
\end{proof}

\noindent Core enforces A4a via online $D \in \mathcal{R}$ check.
Theorem~\ref{thm:transcript-auth} adds that result-set authorization is verifiable from
$\tau_\text{audit}$ alone (A4b).

\subsection{End-to-End Security}

We now state a composite security theorem that unifies P1--P3.

\begin{definition}[SAP Security Game]
\label{def:sap-game}
An adversary $\mathcal{A}$ interacts with:
$\mathcal{O}_\text{Issue}$ (obtain fresh $(S, N, t_{\exp})$),
$\mathcal{O}_\text{Search}$ (submit tokens, receive
$\mathcal{R}$, $\mathit{root}$, $\rho$),
$\mathcal{O}_\text{Prove}$ (request honest proofs for chosen
$(S, D)$), and
$\mathcal{O}_\text{Verify}$ (submit $(\tau, S)$ for online
check). $\mathcal{A}$ sees all oracle outputs and network
messages. $\mathcal{A}$ wins if it achieves any of:

\smallskip\noindent\textbf{Win$_1$ (Session rebinding).}
$\mathcal{A}$ produces $\tau$ such that
$\text{Verify}(S, \tau) = 1$ but $\tau$ was generated using
$N_{S'} \neq N_S$.

\smallskip\noindent\textbf{Win$_2$ (Unauthorized drop).}
$\mathcal{A}$ produces $\tau$ for drop $D$ such that
$\text{Verify}(S, \tau) = 1$ but $D \notin \mathcal{R}_S$.

\smallskip\noindent\textbf{Win$_3$ (Audit re-association).}
$\mathcal{A}$ produces $\tau^*_\text{audit}$ such that
$\text{Audit}(\tau^*_\text{audit}) = 1$ but $\pi$ was generated
under a different authorization context from what
$\tau^*_\text{audit}$ claims.
\end{definition}

\begin{theorem}[SAP Security of Full \sbpp{}]
\label{thm:sap}
Under SHA-256 collision resistance, Groth16 soundness, LP
injectivity, Merkle target-binding security, and EUF-CMA security
of the receipt signature scheme, the advantage of any PPT
$\mathcal{A}$ in the SAP security game
(Definition~\ref{def:sap-game}) is negligible:
\[
\Pr[\text{Win}_1 \lor \text{Win}_2 \lor \text{Win}_3] \leq
  \negl(\lambda).
\]
\end{theorem}

\begin{proof}[Proof sketch]
Win$_1$: given $\mathcal{A}$ that wins, we build a SHA-256
collision finder or Groth16 soundness breaker
(Theorems~\ref{thm:core-binding}--\ref{thm:cross-session}).
Win$_2$: given $\mathcal{A}$ that wins, we build a Merkle
second-preimage finder (Theorem~\ref{thm:transcript-auth}).
Win$_3$: $\mathcal{A}$ must either forge the receipt signature
(contradicting EUF-CMA) or produce a valid Merkle witness for
$D \notin \mathcal{R}$ under the signed root (contradicting
Merkle security), or produce a digest collision (contradicting
SHA-256). Each reduction loses at most a polynomial factor.
Full reductions are in Appendix~\ref{app:sap-proof}
(Win$_1$ $\to$ Theorem~5, Win$_2$ $\to$ Theorem~6,
Win$_3$ $\to$ Theorem~7).
\end{proof}

\subsection{Leakage Analysis}

Following the SSE leakage taxonomy~\cite{sse-leakage}, plaintext
search reveals the user's geohash cell directly. Encrypted search
(GridSE/\sbpp{}) reveals only the access pattern (which $k$ tokens
matched)---qualitatively different (content-based vs.\
location-based) but not zero~\cite{leaker-sok,leakage-inversion,quantifying-location-privacy}.
Per-query, \sbpp{} and GridSE have identical SSE leakage.
\sbpp{}'s improvement is in \emph{cross-phase} leakage: session
binding (Theorem~\ref{thm:cross-session}) prevents re-association of proofs with
different search sessions, and session consumption prevents replay.

\smallskip\noindent\textbf{Residual leakage.}
SSE access patterns (deterministic HMAC; mitigated by
ORAM~\cite{oram}), volume patterns, and timing correlation remain.
These are inherent to the SSE layer and orthogonal to \sbpp{}.

\section{Implementation}
\label{sec:impl}

We implement \sbpp{} in Zairn~\cite{zairn2026} ($\sim$740 new/modified
LOC) by reusing the existing encrypted search module and Groth16
circuit (474 constraints, BN128) without modification. All fields
use length-prefixed UTF-8 encoding with domain separators
(\texttt{SBPP-v1}, \texttt{SBPP-LEAF}, \texttt{SBPP-NODE}).
SHA-256 outputs are mapped to the BN128 scalar field by modular
reduction. Cross-language test vectors and canonicalization
details are provided in the open-source release~\cite{zairn2026}.
The implementation includes 151 unit and integration tests.

\section{Evaluation}
\label{sec:eval}

We address five research questions:
\textbf{RQ1}~leakage profile vs.\ plaintext/GridSE;
\textbf{RQ2}~attack prevention vs.\ naive composition;
\textbf{RQ3}~binding correctness;
\textbf{RQ4}~protocol-path overhead;
\textbf{RQ5}~scalability, mobile latency, and deployment properties.

\subsection{Setup}

Synthetic drops are distributed uniformly at random in the Tokyo
metropolitan area (35.6--35.8$^\circ$N, 139.6--139.9$^\circ$E).
We additionally evaluate under a clustered distribution and under
\textbf{real-world data}: 110{,}776 OpenStreetMap POIs (amenities,
shops, tourism, leisure) in the same bounding box, fetched via
the Overpass API. All experiments run on Node.js v24 on a consumer
laptop (x86-64). Scripts are in the open-source
repository~\cite{zairn2026}.

\subsection{RQ1: Leakage Profile}

Per-query SSE leakage (search pattern, access pattern, volume
pattern~\cite{sse-leakage}) is \emph{identical} for GridSE and
\sbpp{}---the same tokens are sent. Plaintext search additionally
reveals the user's geohash cell. \sbpp{}'s improvement is
orthogonal: \emph{authorization provenance} (A1/A2/A4b in
Table~\ref{tab:attack}). Cross-query Jaccard similarity is
$>$0.98 for all schemes (inherent to geohash).

\subsection{RQ2: Naive Composition vs.\ \sbpp{}}

A key question is whether the session binding must be embedded in
the proof's public inputs, or whether simpler alternatives suffice.
We compare nine protocol variants against six attack classes
($N = 100$ trials each):

\begin{itemize}
\item[\textbf{V1}] Plaintext search + ZKP (no encrypted search,
  no session binding).
\item[\textbf{V2}] GridSE + ZKP (encrypted search, no session
  binding in proof).
\item[\textbf{V3}] GridSE + app-layer nonce (server checks nonce
  separately, but nonce is \emph{not} in the proof digest).
\item[\textbf{V4a}] \sbpp{} Core (session nonce in proof digest;
  server also performs $D \in \mathcal{R}$ check).
\item[\textbf{V4b}] \sbpp{} Full (nonce + Merkle root in proof;
  transcript-level authorization via membership proof).
\item[\textbf{V5}] Signed capability (server signs sidecar token
  binding session + drop; not in proof).
\item[\textbf{V6}] Per-drop signed permit (authorization-only
  baseline; no ZKP---does not prove proximity).
\item[\textbf{V7}] MAC-bound result authorization (server computes
  MAC over result set; MAC key is server secret).
\item[\textbf{V8}] Internalized token hash (hash of signed
  capability token committed in \texttt{pub[7]}; strongest
  proof-internalized baseline).
\end{itemize}

Table~\ref{tab:attack} reports the results.

\begin{table}[!t]
\caption{Attack Matrix ($N$\,=\,100 trials). \checkmark\ = blocked;
$\times$\ = succeeds. V5--V7: proof-external; V8: internalized
token hash (strongest baseline).\label{tab:attack}}
\centering\footnotesize
\resizebox{\columnwidth}{!}{%
\begin{tabular}{lccccccccc}
\toprule
\textbf{Attack} & \textbf{V1} & \textbf{V2} & \textbf{V3} & \textbf{V4a} & \textbf{V4b} & \textbf{V5} & \textbf{V6} & \textbf{V7} & \textbf{V8} \\
\midrule
A1: Cross-session   & $\times$ & $\times$ & $\times$ & \checkmark & \checkmark & \checkmark & \checkmark & \checkmark & \checkmark \\
A2: Re-association  & $\times$ & $\times$ & $\times$ & \checkmark & \checkmark & \checkmark & \checkmark & \checkmark & \checkmark \\
A3: Cross-drop      & \checkmark & \checkmark & \checkmark & \checkmark & \checkmark & \checkmark & \checkmark & \checkmark & \checkmark \\
A4a: Result-set (online) & $\times$ & $\times$ & $\times$ & \checkmark & \checkmark & \checkmark & \checkmark & \checkmark & \checkmark \\
A4b: Auth.\ provenance & $\times$ & $\times$ & $\times$ & $\times$ & \checkmark & $\times$ & $\times$ & $\times$ & \checkmark$^*$ \\
A5: Delayed replay  & $\times$ & $\times$ & \checkmark & \checkmark & \checkmark & \checkmark & \checkmark & \checkmark & \checkmark \\
\bottomrule
\multicolumn{10}{l}{\scriptsize $^*$V8 passes A4b \emph{if} the token
contains equivalent content (session nonce + result-set root).} \\
\end{tabular}%
}
\end{table}

A1--A4a and A5 are prevented by all session-aware variants via
server-side checks. The \textbf{critical row is A4b}. Among
proof-external designs (V5--V7), none passes: re-association
succeeds because proofs are session-agnostic.

\textbf{V8 (internalized token hash)} commits an opaque
signed-token hash in pub[7], achieving A4b when the token includes
equivalent content. V8 validates the SAP principle that
internalization is necessary; the comparison axis between V8 and
\sbpp{} is therefore not security superiority but \emph{audit
semantics}: V8 couples authorization evidence into one hash,
so an auditor must reconstruct the full token to check any single
property. \sbpp{} decomposes the committed content into three
independently auditable components (nonce~$\to$~P1, Merkle
root~$\to$~P2, receipt~$\to$~P3), enabling property-level fault
isolation, $O(1)$ compact-mode state, and verifier-state
discipline (the server never stores $\mathcal{R}$ in Full mode).

\subsection{RQ3: Binding Validation}

All eleven scenarios produce outcomes consistent with the
design requirements and
Theorems~\ref{thm:core-binding}--\ref{thm:sap} (100/100 trials
each). These are implementation sanity checks; security rests on
the cryptographic assumptions.

\noindent\textbf{Fault-localization validation (V8 vs.\ \sbpp{}).}
We inject three fault types into audit records ($N = 100$ each):
session rebinding (swap nonce), result-set tampering (alter root),
and fabricated context (forge receipt signature). Full \sbpp{}
localizes each: nonce mismatch (100/100), Merkle witness failure
(100/100), and receipt signature failure (100/100) respectively.
V8 returns undifferentiated ``hash mismatch'' in all three cases
(100/100)---confirming the fault-isolation advantage in
Table~\ref{tab:v8}.

\subsection{RQ4: Protocol-Path Overhead}

Table~\ref{tab:perf} reports median latency over 1{,}000 iterations
with 100-iteration warmup.

\begin{table}[!t]
\caption{Protocol-Path Latency ($\mu$s, 1{,}000 drops; excludes
Groth16 proving, which is identical across schemes)\label{tab:perf}}
\centering
\begin{tabular}{lcccc}
\toprule
\textbf{Scheme} & \textbf{Med.} & \textbf{Mean} & \textbf{P95} & \textbf{P99} \\
\midrule
Plaintext      & 2.9   & 13.1  & 55.3  & 94.8  \\
GridSE         & 390.8 & 410.6 & 567.2 & 647.2 \\
\sbpp{}        & 417.9 & 441.2 & 641.7 & 764.1 \\
\bottomrule
\end{tabular}
\end{table}

\sbpp{} adds 27\,$\mu$s median over GridSE (+7\%), broken down as:
session creation 17\,$\mu$s, challenge digest 1.9\,$\mu$s, session
validation $<$1\,$\mu$s. The dominant cost in both GridSE and
\sbpp{} is HMAC token generation (359\,$\mu$s), which is identical.
Both encrypted schemes are $\sim$140$\times$ slower than plaintext
geohash encoding (2.9\,$\mu$s), but remain under 0.5\,ms
median---well within interactive latency budgets.

\noindent\textbf{Important caveat.}
The 7\% overhead is for the \emph{protocol path only} (session
management, token generation, digest computation), not for the
full unlock flow. Groth16 proof generation (42--125\,ms warm on
mobile; Table~\ref{tab:mobile}) dominates the user-perceived
latency and is identical across all schemes. The full
end-to-end overhead of \sbpp{} over GridSE is thus
$\sim$0.03\,ms on a $\sim$125\,ms baseline ($<$0.03\%).
\subsection{RQ5: Scalability and Real-World Distribution}

Token matching scales linearly: 0.04\,ms at 1{,}000 drops to
5.2\,ms at 100{,}000 (radius 1\,km). Precision is moderate
(0.45--0.61) due to geohash coarseness; recall is $>$0.92.
Under 110{,}776 OpenStreetMap POIs (Tokyo bounding box), match
count is $\sim$2.7$\times$ higher than uniform due to real-world
clustering, but precision improves (0.60 vs.\ 0.54) and recall
stays $>$0.93. Per-query SSE leakage is identical for GridSE
and \sbpp{} under all distributions.

\subsection{Session Consumption Atomicity}

Four scenarios ($N = 1{,}000$ each): sequential double-submit
(second always rejected), parallel double-submit (exactly one
accepted), expiry boundary (correct in 100\%), bulk
10{,}000-session lifecycle (exact counts). The in-memory store
is single-threaded; multi-instance deployments require
database-level atomic guards.

\subsection{Mobile End-to-End Latency}

Table~\ref{tab:mobile} reports warm-run results ($n = 30$) on
Xiaomi POCO X7 Pro (Android, Chrome~146, 8-core, 8\,GB) and
iPhone~16~Pro (iOS~18.7, Safari~26.2).

\begin{table}[!t]
\caption{Mobile Proof Latency (warm, $n = 30$)\label{tab:mobile}}
\centering\footnotesize
\begin{tabular}{lcc}
\toprule
& \textbf{POCO X7 Pro} & \textbf{iPhone 16 Pro} \\
\midrule
Prove median (ms)  & 125 & 42  \\
Prove mean (ms)    & 128 & 43  \\
Prove P95 (ms)     & 181 & 51  \\
Prove cold (ms)    & 340 & 192 \\
Verify median (ms) & 26  & 9   \\
Payload (B)        & 936 & 935 \\
\bottomrule
\end{tabular}
\end{table}

Warm proof generation is 42--125\,ms median; including \sbpp{}
overhead (0.4\,ms), the complete unlock is under 130\,ms.
Cold-start penalties (192--340\,ms) reflect WASM compilation.
Safari is $\sim$3$\times$ faster than Chrome for both phases.

\subsection{Merkle State Compression}

Table~\ref{tab:merkle} compares per-session state and verification
overhead for the stateful and Merkle-root verifier modes.

\begin{table}[!t]
\caption{Merkle-Root Verifier Scaling\label{tab:merkle}}
\centering\footnotesize
\begin{tabular}{rcccc}
\toprule
$|\mathcal{R}|$ & \textbf{Build (ms)} & \textbf{Prove ($\mu$s)} & \textbf{Verify ($\mu$s)} & \textbf{Steps} \\
\midrule
100      & 0.2  & 20   & 11    & 7  \\
1\,000   & 1.6  & 25   & 19    & 10 \\
5\,000   & 5.2  & 29   & 122   & 13 \\
10\,000  & 12.4 & 60   & 227   & 14 \\
20\,000  & 27.6 & 68   & 460   & 15 \\
50\,000  & 54.2 & 133  & 1137  & 16 \\
\bottomrule
\end{tabular}
\end{table}

Compact-mode state is 104\,B regardless of $|\mathcal{R}|$
(vs.\ $\sim$11\,kB stateful at 1{,}000). Tree construction scales
to 54\,ms at 50{,}000. Merkle verification scales from
19\,$\mu$s to 1.1\,ms---negligible vs.\ Groth16 (9--26\,ms).

\subsection{Offline Audit Replay}

100 sessions: build Merkle trees, generate audit records, purge
session state, re-verify from $\tau_\text{audit}$ alone. All 100
pass. The \sbpp{}-specific audit overhead (receipt signature
check, digest recomputation, Merkle membership) is
$\sim$4.4\,$\mu$s per record; Groth16 proof re-verification adds
9--26\,ms (Table~\ref{tab:mobile}), dominating the audit cost.
Under Core \sbpp{} (V4a), audit fails---no committed
root---confirming A4b requires Full \sbpp{}.

\subsection{Cross-Session Re-Association}

Over 1{,}000 sessions, an adversary who matches
$(D, \text{pv}, e)$ across session boundaries re-associates
proofs in 71.8\% of cases for V2/V3 (below 100\% because epoch
rotation reduces the eligible pool). Under \sbpp{} (V4a/V4b),
the success rate is 0\%---the session nonce in pub makes each
proof session-specific. This is \sbpp{}'s primary cross-phase
improvement; per-query SSE leakage is unchanged.

\noindent\textbf{Payload and audit storage.}
Full mode adds $\sim$320\,B per unlock (Merkle path); audit
records total $\sim$1.5\,kB each ($\sim$45\,GB/month at 1\,M
unlocks/day).

\subsection{Concurrency, Malicious Server, Isolation}

\noindent\textbf{Atomic consumption.}
Under DB-level \texttt{WHERE consumed=false} guards
($N = 1{,}000$), concurrent double-submit achieves exactly-one
acceptance in 100\% of trials.

\smallskip\noindent\textbf{Malicious server impact.}
We test five malicious behaviors ($N = 100$ each): candidate
omission (detectable with honest reference root, 100\%); biased
root signing (\emph{not} detectable---receipt authenticates what
was issued, not correctness); predictable nonce (enables cross-user
transfer; random nonces prevent this); session refusal
(\emph{not} detectable---DoS); authorization forgery
(\emph{impossible}---0\% forgery rate). These confirm P1--P3 under
HBC server while showing the completeness/fairness boundary.

\smallskip\noindent\textbf{Multi-client isolation.}
100\,k sessions in $\sim$1\,s ($\sim$19.5\,MB). Concurrent
issue+verify ($K = 10$--500) all correct. Cross-session isolation:
0/9{,}900 cross-access attempts succeeded.

\section{Discussion}
\label{sec:discussion}

\subsection{What \sbpp{} Protects and What It Does Not}

Core \sbpp{} prevents cross-session proof substitution (A1/A2) and
enables online result-set authorization (A4a). Full \sbpp{} lifts
authorization to the transcript level (A4b), enabling offline
auditability. Both levels prevent delayed replay (A5) via session
TTL and consumption.

Core ensures proofs belong to their originating session; Full
additionally commits the authorization context via Merkle root and
witness. Combined with signed receipts, this yields replayable
authorization provenance under HBC server.

\sbpp{} does \emph{not} address:
\begin{itemize}
\item \textbf{SSE access pattern leakage.} Deterministic tokens
  reveal whether two queries target the same cell. This is
  inherent to SSE; volume-hiding SSE or ORAM~\cite{oram,path-oram} would
  mitigate it at significant performance cost.
\item \textbf{Location spoofing.} \sbpp{} assumes the client's
  coordinates are genuine. Spoofing resistance requires
  orthogonal mechanisms (trust scoring, multi-sensor
  fusion~\cite{guardian-gps,secure-positioning}).
\item \textbf{Malicious server.} \sbpp{} assumes an
  honest-but-curious server. P1--P3 guarantee that accepted proofs
  are bound to the authorization context the server \emph{issued},
  but not that the server issued \emph{honestly} (e.g., it could
  omit candidates or sign a biased root). Extending to malicious
  servers requires verifiable query processing~\cite{ads-index}.
\end{itemize}

\subsection{Why Not Simpler Alternatives?}

\noindent\textbf{Internalized token hash (V8).}
V8 commits an opaque signed-token hash in \texttt{pub[7]},
validating the SAP principle that internalization is necessary.
However, V8 is an ad-hoc instantiation that couples all
authorization evidence into a single hash, whereas \sbpp{}
provides the \emph{practical decomposition} of SAP requirements
into independently verifiable components. This decomposition matters for forensic practice.
\emph{Example}: if an audit record fails, \sbpp{} localizes the
failure---a nonce mismatch indicates session rebinding (A1/A2),
a Merkle witness failure indicates result-set escape (A4b), and
a receipt signature failure indicates fabricated context (P3
violation). With V8's opaque hash, all three failure modes produce
the same symptom (hash mismatch), and the auditor cannot
distinguish them without reconstructing the full token.
The security contribution of this paper is the SAP notion and the
proof-external separation; \sbpp{} is a practical realization
that enables property-level fault isolation in audit.

\begin{table}[!t]
\caption{V8 vs.\ Full \sbpp{}: Audit Properties\label{tab:v8}}
\centering\footnotesize
\begin{tabular}{lcc}
\toprule
& \textbf{V8 (opaque)} & \textbf{Full \sbpp{}} \\
\midrule
A4b (provenance)     & \checkmark$^*$ & \checkmark \\
Fault localization   & No & Yes \\
Auditor prerequisites & Token format & Public key only \\
Token-format changes  & Breaks audit & Transparent \\
\midrule
\multicolumn{3}{l}{\textit{Failure-mode diagnostics}} \\
\midrule
Session rebinding (A1/A2) & hash mismatch & nonce mismatch \\
Result-set escape (A4b)   & hash mismatch & Merkle failure \\
Fabricated context (P3)   & hash mismatch & sig.\ failure \\
\bottomrule
\multicolumn{3}{l}{\scriptsize $^*$If token contains equivalent content.}
\end{tabular}
\end{table}

\noindent\textbf{Multi-instance deployment.}
Session consumption in the prototype uses an in-memory store
(single-threaded). In production, atomic consumption maps to a
database \texttt{UPDATE ... WHERE consumed = false} or a Redis
compare-and-swap; the Merkle root and receipt are stateless and
require no cross-instance coordination.

\noindent\textbf{Toward stronger privacy.}
Volume-hiding search and ORAM-based retrieval~\cite{oram,path-oram}
would reduce access-pattern leakage orthogonally to \sbpp{}.

\section{Related Work}
\label{sec:related}

\textbf{Encrypted geographic search and leakage.}
GridSE~\cite{gridse} introduced prefix-based SSE for geographic
queries; later work explored geometric range search and practical
private range search~\cite{hier-sse,range-sse}. SSE leakage and its
exploitation are also well-studied~\cite{curtmola-sse,sse-leakage,leaker-sok,leakage-inversion};
recent work studies privacy-preserving and verifiable spatial queries
over encrypted data~\cite{miao-spatial-query,verifiable-geo-query}.
\sbpp{} is complementary: it preserves per-query SSE leakage while
binding search authorization to a later proof.

\textbf{Location proofs and authentication.}
Groth16~\cite{groth16} underlies pairing-based ZKP proximity proofs.
Prior location-proof and location-authentication schemes typically
assume trusted witnesses, infrastructure, or trusted location
sources~\cite{proving-location,veriplace,applaus,stamp,locproof-survey,loc-auth-sigmod,loc-auth-wpes,private-proximity-tags,secure-positioning}.
\sbpp{} instead binds an \emph{encrypted search session} to a
ZKP-based unlock flow.

\textbf{Authenticated data structures and verifiable composition.}
Full \sbpp{}'s Merkle-root commitment is related to authenticated
data structures~\cite{merkle-signature,ads-generic,ads-dynamic,ads-index}.
More broadly, \sbpp{} addresses the composition problem of binding
a query result to a subsequent cryptographic action---a pattern
that arises whenever search and verification are separate protocol
phases. Spatial cloaking and differential privacy address
orthogonal query-side privacy goals~\cite{k-anon-loc,dp-loc,quantifying-location-privacy}.

\section{Conclusion}
\label{sec:conclusion}

We introduced the \emph{search-authorized proof} (SAP) security
notion for systems that compose encrypted geographic search with
zero-knowledge proximity verification: three properties (P1--P3)
that formalize when an audit transcript (proof, receipt, and
Merkle witness) can forensically attest the proof's originating
authorization. We exhibited an audit re-association
attack showing that proof-external mechanisms leave this boundary
open, and established design requirements for the committed content
required by each property.

The principal scientific contributions are the SAP notion and the
proof-external separation; \sbpp{} is a practical
realization, decomposing the required committed content into a
session nonce (P1), a Merkle root (P2), and a signed receipt (P3)
via a digest-based binding chain requiring no circuit modification.
Protocol-path overhead is 7\% in a single-instance prototype;
DB-backed atomic consumption and multi-client isolation were
validated in targeted experiments, though production-scale
multi-instance deployment remains future work.

\sbpp{} operates under an honest-but-curious server and guarantees
that accepted proofs are bound to the authorization context the
server \emph{issued}---not that the server issued
\emph{honestly}. Search completeness, fairness, and SSE
access-pattern leakage are outside its scope. Extending to
malicious servers via verifiable query processing, and reducing
per-query leakage via ORAM, are key directions for future work.
Artifacts are available at~\cite{zairn2026}.

\section*{AI Usage Disclosure}
Generative AI tools (OpenAI ChatGPT and Anthropic Claude) were
used for coding assistance, literature search, and editorial
refinement during manuscript preparation. All AI-assisted outputs
were reviewed, edited, and validated by the author, who takes full
responsibility for the content of the manuscript.

\section*{Declaration of Competing Interest}
The author declares no competing interests.

\appendices
\section{Proof of Theorem~\ref{thm:sap} (SAP Security)}
\label{app:sap-proof}

\noindent\textbf{Audit algorithm.}
$\text{Audit}(\pi, \text{pub}, D, \mathit{mp}, \rho)$:
(1)~verify $\text{Sig}_\text{srv}(\rho)$;
(2)~extract $(N, \mathit{root}, \text{pv}, e)$ from $\rho$;
(3)~recompute $cd' \gets H(\LP(\texttt{SBPP-v1}, D, \text{pv},
e, N, \mathit{root}))$; check $\text{pub}[7] = cd'$;
(4)~verify Merkle path $\mathit{mp}$ for leaf
$H(\LP(\texttt{SBPP-LEAF}, D))$ against $\mathit{root}$;
(5)~verify $\pi$ over $\text{pub}$ via Groth16.
Return 1 iff all checks pass.

\smallskip\noindent\textbf{Notation.}
``Different authorization context'' means $N' \neq N$
(\emph{session difference}) or $\mathit{root}' \neq \mathit{root}$
(\emph{result-set difference}).

\begin{theorem}[Win$_1$ Reduction]
If $\mathcal{A}$ wins Win$_1$, we construct a SHA-256 collision
finder $\mathcal{B}$ (Groth16 soundness ensures $\pi^*$ fixes
$\text{pub}^*[7]$, which is the starting point of the reduction).
\end{theorem}
\begin{proof}
$\mathcal{A}$ outputs $\tau^* = (\pi^*, \text{pub}^*, D)$
accepted by $\text{Verify}_\text{full}(S, \tau^*)$ (nonce $N_S$),
but $\pi^*$ was generated with $N' \neq N_S$. Verification
requires $\text{pub}^*[7] = cd_\text{full}(S, D)$, which embeds
$N_S$. By Groth16 soundness, $\pi^*$ fixes $\text{pub}^*[7]$.
Since $\pi^*$ was generated with $N'$, we have
$\text{pub}^*[7] = H(\LP(\ldots, N', \ldots))$. Acceptance
gives $H(\LP(\ldots, N_S, \ldots)) = H(\LP(\ldots, N', \ldots))$.
LP injectivity ensures the preimages differ: $\mathcal{B}$ outputs
this pair as a SHA-256 collision.
\end{proof}

\begin{theorem}[Win$_2$ Reduction]
If $\mathcal{A}$ wins Win$_2$, we construct $\mathcal{B}$ that
breaks Merkle authentication security (target-binding: given a
root $r$ and a set $\mathcal{R}$, find $D \notin \mathcal{R}$ and
a path $\mathit{mp}$ such that $\mathit{mp}$ verifies for $D$
against $r$).
\end{theorem}
\begin{proof}
$\mathcal{A}$ outputs $(\pi, \text{pub}, D, \mathit{mp})$ accepted
under $S$ with $D \notin \mathcal{R}_S$. The honest tree built
over $\mathcal{R}_S$ does not contain leaf
$H(\LP(\texttt{SBPP-LEAF}, D))$, yet $\mathit{mp}$ verifies
against $\mathit{root}_S$. $\mathcal{B}$ forwards
$(D, \mathit{mp}, \mathit{root}_S, \mathcal{R}_S)$ to the
Merkle challenger, breaking target-binding.
\end{proof}

\begin{theorem}[Win$_3$ Reduction]
If $\mathcal{A}$ wins Win$_3$, we construct either an EUF-CMA
forger, a SHA-256 collision finder, or a Merkle target-binding
breaker.
\end{theorem}
\begin{proof}
$\mathcal{A}$ outputs $\tau^*_\text{audit}$ passing all five Audit
checks, but $\pi$ was generated under $A' \neq A$.

\emph{Case (a): session difference} ($N' \neq N$). Step~(1)
verifies the receipt signature on $\rho$. Either $\rho$ was forged
($\mathcal{B}$ breaks EUF-CMA) or $\rho$ is honest and
authenticates $N$. Step~(3) recomputes $cd'$ with $N$; since $\pi$
was generated with $N'$, acceptance yields a SHA-256 collision (as
in Win$_1$).

\emph{Case (b): result-set difference}
($\mathit{root}' \neq \mathit{root}$). Either $\rho$ is forged
(EUF-CMA) or the digest equality yields a collision. If
$D \notin \mathcal{R}$ under the honest root, step~(4) yields a
Merkle target-binding break (as in Win$_2$).
\end{proof}

\noindent By a union bound,
$\Pr[\text{Win}_1 \lor \text{Win}_2 \lor \text{Win}_3] \leq
\text{Adv}^\text{CR}_\text{SHA-256} +
\text{Adv}^\text{Sound}_\text{Groth16} +
\text{Adv}^\text{Merkle}_\text{target-bind} +
\text{Adv}^\text{EUF-CMA}_\text{Sig}$,
all negligible under standard assumptions. \qed


\end{document}